\documentclass[lettersize,journal]{IEEEtran}
\usepackage{amsmath,amsfonts}
\usepackage{amsthm}
\usepackage{algorithmic}
\usepackage{algorithm}
\usepackage{array}
\usepackage[caption=false,font=normalsize,labelfont=sf,textfont=sf]{subfig}
\usepackage{textcomp}
\usepackage{stfloats}
\usepackage{url}
\usepackage{verbatim}
\usepackage{graphicx}
\usepackage{mathtools, nccmath}
\usepackage{cite}
\usepackage{balance}
\newtheorem{theorem}{Theorem}[section]

\begin{document}

\title{On the required radio resources for ultra-reliable communication in highly interfered scenarios}

\author{Gilberto Berardinelli, Ramoni Adeogun
\thanks{Gilberto Berardinelli and Ramoni Adeogun are with the Department of Electronic Systems, Aalborg University, Denmark (e-mail:\{gb, ra\}@es.aau.dk)}}

\markboth{Journal of \LaTeX\ Class Files,~Vol.~14, No.~8, August~2021}%
{Shell \MakeLowercase{\textit{et al.}}: A Sample Article Using IEEEtran.cls for IEEE Journals}


\maketitle

\begin{abstract}
Future wireless systems are expected to support mission-critical services demanding higher and higher reliability. In this letter, we dimension the radio resources needed to achieve a given failure probability target for ultra-reliable wireless systems in high interference conditions, assuming a protocol with frequency hopping combined with packet repetitions. We resort to packet erasure channel models and derive the minimum amount of resource units in the case of receiver with and without collision resolution capability, as well as the number of packet repetitions needed for achieving the failure probability target. Analytical results are numerically validated, and can be used as benchmark for realistic system simulations.
\end{abstract}


\section{Introduction}
\IEEEPARstart{T} {\MakeLowercase{he}} increasing utilization of wireless for mission-critical applications has stressed the importance of ensuring reliable over-the-air communication \cite{mohammed2019mission}.
The 5th generation (5G) radio access technology, standardized by the third generation partnership project (3GPP), has introduced the concept of ultra-reliable low latency communication (URLLC), targeting reliability levels, defined as packet error rates, in the order of $10^{-5}-10^{-6}$ \cite{liu2020analyzing}. The trend towards wireless support of critical services is expected to continue in the coming years with 6th Generation (6G) radios, targeting more and more demanding reliability and availability requirements \cite{viswanathan2020communications}.


For example, the recent concept of 6G \emph{in-X subnetworks} envisions short-range radio cells to be installed inside entities like industrial robots, production modules, vehicles or even human bodies, for critical control applications \cite{berardinelli2021extreme}. Specifically, in-X subnetworks may replace hard real time industrial communication protocols such as PROFINET Isochronous Real Time (IRT) and Ethernet Powerlink, characterized by strictly periodic transmissions with deterministic cycle times \cite{schumacher2008new}. In these systems, a certain packet transmission by e.g., a sensor is expected to reach the receiving node within a bounded time, for the sake of preserving the stability of the control loop. Since in-X subnetworks can become very dense, such as those installed in robots operating in a matrix production mode, interference is a major threat hindering the reliability of the communication. 

In our previous work, we have proposed a transmission protocol featuring a combination of packet repetitions (a.k.a. repetition code) and frequency channel hopping, as a viable approach for achieving ultra-reliable communication in dense subnetwork scenarios characterized by a large number of uncoordinated persistent transmitters concurring for the same radio resources \cite{adeogun2020towards}. In particular, each packet transmission is repeated a number of times over extremely short time intervals, in order to ensure ultra-reliable communication with sub-ms latencies in dense scenarios. Channel hopping is a well-known approach for harvesting frequency and interfering diversity; our protocol is inspired to the main operational principle of the IEEE 802.15.1 standard \cite{aalam2011ieee}, as well as the $\rm{K}$-repetition mode, proposed by the 3GPP as a grant-free alternative to traditional grant-based access \cite{le2021enhancing}. 

The performance of our proposed protocol has been evaluated via system level simulations in an industrial scenario (e.g.,\cite{adeogun2020towards, adeogun2020distributed}), highlighting the complex trade-offs between resource utilization, number of repetitions and deployment characteristics. Though system level simulations are a powerful approach for realistic performance analysis, their results are necessarily conditioned to specific assumptions in terms of operational scenario and radio propagation model (e.g., an indoor factory propagation scenario in our analysis). This makes difficult to generalize the learnings and distinguish the protocol aspects from the combined effects of radio channel and deployment characteristics. 

In this letter, we perform an analysis of our frequency-hopped protocol with focus on the relationship between resource utilization, packet repetitions, and high reliability, for the case of a large number of persistent radio devices sharing the same resource grid. We resort to erasure channel models to derive analytically the minimum amount of radio resources for supporting a target failure probability, and the corresponding required number of packet repetitions. The usage of packet erasure channel models is a very common approach for the basic analysis of wireless communication protocols (see, e.g., the works on coded slotted ALOHA and coded random access in \cite{sun2017coded,paolini2015coded}), as they allow to isolate the pure protocol aspects from the radio propagation and modelling parts. Our aim is to obtain closed form expressions for the number of resource units ensuring a given failure probability target, to be used as a benchmark for coming system level analysis of 6G frequency-hopped systems for ultra-reliable communication. The validity of our analytical results is proved via a numerical study.   


\section{Setting up the scene}
We consider $d+1$ devices connected to one or more access points (APs). We assume one out of the $d+1$ devices to be our device of interest, while the other $d$ devices are potential interferers. This can be the case, for example, of a sensor in a subnetwork reporting periodic measurements to its serving APs, with the other $d$ devices belonging to different subnetworks. The reference scenario is pictorially depicted in Fig. \ref{fig1}. We assume for simplicity that the devices adhere to a common clock and are frame aligned, but their operations are not coordinated as each subnetwork operates independently. 

Devices operate over the same radio resources, where each transmission happens over a frame consisting of a grid of $N_{RU}=pq$ resource units, with $p$ frequency channels and $q$ time slots. Obviously, the wireless system designer must dimension the size of a resource unit according to the expected desired signal strength, data packet size and spectrum availability. Each device maps then a data packet in a resource unit, and transmits it $n$ times by selecting for each repetition a different resource unit, for example according to a pseudo-random pattern. It is assumed that each repetition happens over a different time slot and frequency channel, and therefore the constraint $1\leq n\leq {\rm{min}} (p,q)$ holds. 

An example of such transmission is depicted in Fig. \ref{fig2}, where the patterns of two non-interfering devices are shown. The pseudo-random channel hopping is a well-known approach in wireless communications, and has been shown to provide significant reliability gain thanks to the possibility of exploiting time, frequency and interference diversity \cite{kostic2001fundamentals}. In a real system the hopping pattern can change at every transmission, in order to achieve robustness to jammers and malicious interferers \cite{chiarello2021jamming}.


In the next sections, we derive the minimum amount of resources needed for achieving a given transmission failure probability target, considering packet erasure channel models. 

\section{Resource dimensioning with collisions as disruptive events}
We first consider a classic packet erasure channel model, where collisions are disruptive events, and transmission is always successful in the case of at least one collision-free packet repetition \cite{massey1985collision}. In spite of its simplicity, we believe that the packet erasure channel model is well-suited for the basic analysis of an interference-limited in-X subnetwork scenario. This is because in short range subnetworks the receive signal strength of the desired transmitter is expected to be high; moreover, the AP can use techniques such as power control and link adaptation to counteract fading and ensure correct reception for interference-free transmissions \cite{berardinelli2021extreme}. This means, transmission failures are only caused by collisions. 
By using a packet erasure channel model, the $d$ interfering devices are to be intended as those devices whose proximity generate disruptive collisions when operating over the same resource units as the target device. In the following, we use the subscript $0$ to denote the no collision case for the derivation of the failure probability.

Let us consider the transmission of a given packet repetition by the target device. Since an interfering device is transmitting over $n$ of the $pq$ resource units, the probability of the target device having no collision with its transmissions is given by $\frac{pq-n}{pq}$. The probability of having no collision for a given packet repetition with any of the $d$ interfering devices is then
${P_0}=\left(\frac{pq-n}{pq}\right)^d$. 

The complement probability $(1-P_0)$ represents then the probability of having at least one collision at a given packet repetition of the target device. The probability of transmission failure, i.e., having at least one collision for each of the $n$ packet repetitions, is then given by

\begin{equation}
    P_{{f,0}}=\left(1-{P_0}\right)^n=\left(1-\left(1-\frac{n}{N_{RU}}\right)^d\right)^n.
\label{eq:failure0}
\end{equation}

The amount of required resource unit for a given failure probability target $P^{(\rm{T})}_{{f}}$ set by the system designer can be then calculated by setting $P_{f,0}=P^{(\rm{T})}_{{f}}$ and solving (\ref{eq:failure0}) for $N_{RU}$, leading to

\begin{equation}
    {N}_{RU}=\left\lceil{\frac{n}{1-\left(1-{P^{({\rm{T}})}_{{f}}}^{\frac{1}{n}}\right)^{\frac{1}{d}}}}\right\rceil ,
    \label{eq:nru}
\end{equation}

\noindent with $\left\lceil \cdot \right\rceil$ denoting the ceiling operator. 

\begin{figure}
    \centering
    \includegraphics[scale=0.3]{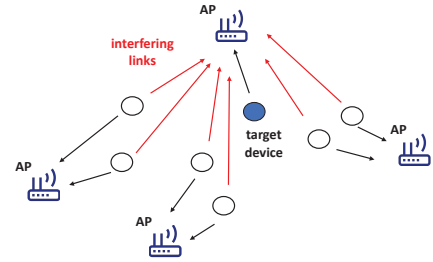}
    \caption{Reference scenario, with a number of devices served by their own AP, interfering the transmissions of the target device.}
    \label{fig1}
\end{figure}

\begin{figure}
    \centering
    \includegraphics[scale=0.34]{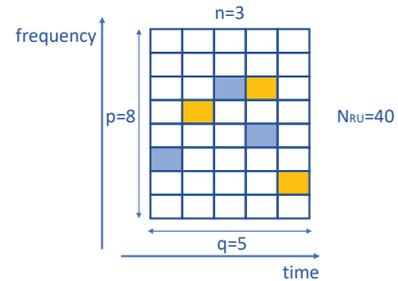}
    \caption{Example of transmission pattern over a frame resource grid. The colors refer to the transmissions of two non-interfering devices.}
    \label{fig2}
\end{figure}

Since  ${N}_{RU}=pq$, the wireless system designer can engineer the split of the resource grid in the time and the frequency domains. Latency-limited services may use a low number $q$ of time slots, while expanding the frequency dimension by increasing the number of frequency channels (calculated as $p={N}_{RU}/q$). Conversely, in the case of limited available spectrum resources, the number $p$ of frequency channels must necessarily be reduced and the number $q$ of time slots increased (calculated as $q={N}_{RU}/p$); this comes at the expense of latency. Also, please observe that the size of a resource unit has to be dimensioned according to the packet size and the achievable signal-to-noise ratio (SNR) at the receiver to ensure correct reception. An analysis of the time/frequency resources needed for achieving a target failure probability as a function of the achievable SNR for the target device is left for future work.

\begin{theorem}
In case the receiver is unable to resolve collisions, for a large number of interfering devices the minimum amount of resource units $N^{min}_{RU}$ for achieving $P^{({\rm{T}})}_{{f}}$ can be approximated as

\begin{equation}
N^{min}_{RU}\approx \left\lceil\frac{{\rm{ln}}\left(P^{({\rm{T}})}_{{f}}\right)}{\left[\left(\frac{1}{2}\right)^{\frac{1}{d}}-1\right]{\rm{ln}}\left(2\right)}\right\rceil. 
\label{eq:nrumin}
\end{equation}
\end{theorem}

\begin{proof}
Let us generalize (\ref{eq:nru}) in the continuous domain as a function of $n$, i.e., $N_{RU}=f(n)$, $n\in \mathbb{R}$, $n\geq 1$. $f(n)$ is a convex function, and attains its minimum at the zero of its derivative. As $f(n)$ is in the form of a quotient $\frac{u}{v}$, its derivative with respect to $n$ will be zero if $vu'-uv'=0$, i.e., \useshortskip

\begin{multline}
    1-\left(1-\left(P^{({\rm{T}})}_{{f}}\right)^{\frac{1}{n}}\right)^{\frac{1}{d}} + \\ +\frac{{\rm{ln}}\left(P^{({\rm{T}})}_{{f}}\right)\left(P^{({\rm{T}})}_{{f}}\right)^{\frac{1}{n}}}{dn} \left(1-\left(P^{({\rm{T}})}_{{f}}\right)^{\frac{1}{n}}\right)^{\frac{1}{d}-1}=0.
    \label{eq:firstproof}
\end{multline}

\noindent By defining $h=\frac{1}{d}$, and dividing both terms of the equality by $h$, (\ref{eq:firstproof}) can be rewritten as \useshortskip

\begin{multline}
    -\frac{\left(1-\left(P^{({\rm{T}})}_{{f}}\right)^{\frac{1}{n}}\right)^{h}-1}{h} + \\ +\frac{{\rm{ln}}\left(P^{({\rm{T}})}_{{f}}\right)\left(P^{({\rm{T}})}_{{f}}\right)^{\frac{1}{n}}}{n} \left(1-\left(P^{({\rm{T}})}_{{f}}\right)^{\frac{1}{n}}\right)^{h-1}=0.
    \label{eq:firstproof2}
\end{multline}

\noindent For large $d$, we have $h \rightarrow 0$, and the leftmost term $\frac{\left(1-\left(P^{({\rm{T}})}_{{f}}\right)^{\frac{1}{n}}\right)^{h}-1}{h}\rightarrow {\rm{ln}}\left(1-\left(P^{({\rm{T}})}_{{f}}\right)^{\frac{1}{n}}\right)$. (\ref{eq:firstproof2}) can be then expressed as \useshortskip

\begin{equation}
    y\rm{ln}\left(y\right)-\left(1-y\right)\rm{ln}\left(1-y\right)=0
    \label{eq:firstproof3}
\end{equation}

\noindent where $y=\left(P^{({\rm{T}})}_{{f}}\right)^{\frac{1}{n}}$. (\ref{eq:firstproof3}) has three zeros, i.e., $y=0, \frac{1}{2}, 1$, with $\frac{1}{2}$ being the solution of interest. We therefore obtain \useshortskip

\begin{equation}
    n=\frac{1}{{\rm{log}}_{P^{({\rm{T}})}_{{f}}}\left(\frac{1}{2}\right)}=-\frac{{\rm{ln}}\left(P^{({\rm{T}})}_{{f}}\right)}{{\rm{ln}}\left(2\right)}.
    \label{eq:firstproof4}
\end{equation}

\noindent This represents the minimum number of packet repetition for achieving $P^{({\rm{T}})}_{{f}}$; interestingly, it does not depend on the number of devices $d$. 
By applying (\ref{eq:firstproof4}) in (\ref{eq:nru}), we obtain the result in (\ref{eq:nrumin}).
\end{proof}

Since for large $d$, $1/d \rightarrow 0$, we can resort to a Maclaurin expansion of the type $y^x\approx 1+x{\rm{ln}}(y)$, and obtain $\left(\frac{1}{2}\right)^{1/d}\approx 1+\frac{1}{d}\rm{ln}\left(\frac{1}{2}\right)$. Therefore:

\begin{equation}
N^{min}_{RU}\approx -\left\lceil\frac{d}{{\rm{ln}}^2\left(2\right)}{\rm{ln}}\left(P^{({\rm{T}})}_{{f}}\right)\right\rceil=-\left\lceil2.0814\cdot d\cdot {\rm{ln}}\left(P^{({\rm{T}})}_{{f}}\right)\right\rceil.
\label{eq:nrumin2}
\end{equation}

\noindent The minimum amount of radio resource units needed for achieving a failure probability $P^{({\rm{T}})}_{{f}}$ has then a logarithmic relationship with the failure probability itself, while it scales linearly with the number of devices.


\section{Dimensioning with resolvable collisions}
Let us assume now that our receiver is able to resolve collisions up to a given extent, thanks to e.g., the usage of successive cancellation receivers and multi-antenna processing \cite{tse2005fundamentals}. In this case, collisions are not disruptive events. Let us denote the number of colliding devices at a given resource unit as $N_{\text{c}}$. When deriving the failure probability, the subscript $c=0, \cdots, N^{\rm{max}}_{\text{c}}$, is used to denote the number of resolvable collisions at a given packet repetition, with $N^{\rm{max}}_{\text{c}}$ being the maximum number of resolvable collisions.

The probability that a packet transmission of the target device collides with one of the $n$ transmissions of a given interfering device is given by $\frac{n}{N_{RU}}$. The probability of having $N_{\text{c}}=c$ collisions in the presence of $d$ interfering devices, with $0\leq c \leq d$, can be then characterized by the probability mass function of a binomial distribution, i.e.,



\begin{equation}
    P(N_{\text{c}}=c)=\binom{d}{c}\left(\frac{n}{N_{RU}}\right)^c\left(1-\frac{n}{N_{RU}}\right)^{d-c}.
\end{equation}

\noindent We are then softening the definition of traditional packet erasure channel models, and assuming that a transmission failure happens when all packet repetitions of the target device experience more than $N^{\rm{max}}_{\text{c}}$ collisions, thanks to the capability of the receiver to resolve at most $N^{\rm{max}}_{\text{c}}$ collisions. The corresponding failure probability can be then expressed as 

\begin{equation}
    P_{f,N^{\rm{max}}_{\text{c}}}=\left(1-\sum^{N^{\rm{max}}_{\text{c}}}_{c=0}P(N_{\text{c}}=c)\right)^n.
    \label{eq:failureNc}
\end{equation}

\noindent Note that (\ref{eq:failureNc}) leads to (\ref{eq:failure0}) for $N^{\rm{max}}_{\text{c}}=0$. The minimum number of radio resource units needed for achieving a certain failure probability target $P^{(\rm{T})}_{{f}}$ can be calculated from  (\ref{eq:failureNc}) by setting  $P_{f,N^{\rm{max}}_{\text{c}}}=P^{(\rm{T})}_{{f}}$ and solving it numerically for $N_{RU}$.

\subsection{Case of a single resolvable collision}
\noindent For the case $N^{\rm{max}}_{\text{c}}=1$, it is possible to derive an explicit expression for $N_{RU}$, as well as for its minimum $N^{min}_{RU}$. 
\begin{theorem}
In case the receiver is able to resolve at most one collision, for a large number of interfering devices the amount of resource units $N_{RU}$ for achieving a target failure probability $P^{(\rm{T})}_{{f}}$ can be approximated as

\begin{equation}
N_{RU}\approx -\left\lceil\frac{nd}{W_{-1}\left(\frac{{P^{(\rm{T})}_{{f}}}^{\frac{1}{n}}-1}{\rm{e}}\right)+1}\right\rceil. 
\label{eq:secondproof0}
\end{equation}

\noindent where $W_{-1}\left(\cdot\right)$ denotes the negative branch of the Lambert W function \cite{corless1996lambertw}. 
\end{theorem}

\begin{proof} 
For $N^{\rm{max}}_{\text{c}}=1$, by setting $P^{(\rm{T})}_{{f}}=P_{f,1}$, (\ref{eq:failureNc}) reads
\begin{equation}
    P^{(\rm{T})}_{{f}}=\left(1-\left(1-\frac{n}{N_{RU}}\right)^d-\frac{nd}{N_{RU}}\left(1-\frac{n}{N_{RU}}\right)^{d-1}\right)^n.
    \label{eq:secondproof}
\end{equation}
Let us generalize (\ref{eq:secondproof}) in the continuous domain of the $N_{RU}$ variable, i.e., $N_{RU}\in\mathbb{R}$. Note that, if $f(N_{RU})=\left(1-\frac{n}{N_{RU}}\right)^d$, then $f'(N_{RU})=\frac{nd}{N^2_{RU}}\left(1-\frac{n}{N_{RU}}\right)^{d-1}$. We therefore obtain $f(N_{RU})+N_{RU}f'(N_{RU})=1-{P^{(\rm{T})}_{{f}}}^{\frac{1}{n}}$. By observing that $f(N_{RU})+N_{RU}f'(N_{RU})=[N_{RU}f(N_{RU})]'$, it must hold that its anti-derivative is of the form $N_{RU}f(N_{RU})=aN_{RU}+b$, with $a=1-{P^{(\rm{T})}_{{f}}}^{\frac{1}{n}}$, and $b$ a constant.

\noindent For large $d$, we can use the approximation $\left(1-\frac{n}{N_{RU}}\right)^d\approx e^{-\frac{nd}{N_{RU}}}$, and (\ref{eq:secondproof}) can be therefore expressed as

\begin{equation}
    N_{RU}e^{-\frac{nd}{N_{RU}}}=aN_{RU}+b.
    \label{eq:secondproof2}
\end{equation}
By setting $w=1/N_{RU}$, (\ref{eq:secondproof2}) can be reformulated as $e^{-ndw}-bw-a=0$. Such equation can be solved by using Lambert W functions \cite[p.4]{edwards2019extension}, leading to

\begin{equation}
    w=-\frac{1}{nd}{\rm{ln}}\left(\frac{b}{nd}W\left(\frac{nde^{\frac{and}{b}}}{b}\right)\right),\label{eq:secondproof3}
\end{equation}

\noindent where $W(x)$ denotes the multivariate Lambert W function. Since for such function it holds that ${\rm{ln}}(W(x))={\rm{ln}}(x)-W(x)$ \cite[p.2]{edwards2019extension}, (\ref{eq:secondproof3}) can be rewritten as 
\begin{multline}
   w=-\frac{1}{nd}\left({\rm{ln}}\left(\frac{b}{nd}\right)+{\rm{ln}}\left(\frac{nde^{\frac{and}{b}}}{b}\right)-W\left(\frac{nde^{\frac{and}{b}}}{b}\right)\right) \\
   =-\frac{and-bW\left(\frac{nde^{\frac{and}{b}}}{b}\right)}{ndb}
\end{multline}
and therefore

\begin{equation}
    N_{RU}\approx-\frac{ndb}{and-bW\left(\frac{nde^{\frac{and}{b}}}{b}\right)}.
    \label{eq:secondproof4}
\end{equation}

\noindent By definition, $W(x)$ has a single real solution equal to $-1$ only at  $x=-\frac{1}{e}$. By setting $\frac{nde^{\frac{and}{b}}}{b}=-\frac{1}{e}$ and $q=nd/b$, we obtain an equation of the form $qe^{aq}+1/e=0$, which is again solvable by using Lambert W functions \cite[p.3]{edwards2019extension}, leading to



\begin{equation}
    b=\frac{nda}{W\left(-\frac{a}{e}\right)}.
    \label{eq:secondproof4a}
\end{equation}

\noindent For $-\frac{1}{e} \leq x<0$, $W(x)$ is two-valued in the positive and negative branch \cite{corless1996lambertw}, i.e., $W_{0}(a/e)$ and $W_{-1}(a/e)$, respectively, with the negative branch solution ensuring positive values for $N_{RU}$ in (\ref{eq:secondproof3}). By inserting (\ref{eq:secondproof4a}) in (\ref{eq:secondproof3}), and applying the ceiling operator, we obtain (\ref{eq:secondproof0}).
\end{proof}

\begin{theorem}
In case the receiver is able to resolve at most one collision, with a large number $d$ of interfering devices the minimum amount of resource units $N^{min}_{RU}$ for achieving a target failure probability $P^{({\rm{T}})}_{{f}}$ can be approximated as

\begin{equation}
    N^{min}_{RU}\approx -\left\lceil 0.7502\cdot d\cdot {\rm{ln}}\left(P^{(\rm{T})}_{{f}}\right)\right\rceil
    \label{eq:thirdproof0}
\end{equation}

\end{theorem}

\begin{proof}
By setting $z=\frac{\left(P^{(\rm{T})}_{{f}}\right)^{\frac{1}{n}}-1}{\rm{e}}$, we obtain $n=\frac{{\rm{ln}}\left(P^{({\rm{T}})}_{{f}}\right)}{{\rm{ln}}\left(1+ez\right)}$. In the continuous domain, (\ref{eq:secondproof0}) becomes then:

\begin{equation}
    N_{RU}\approx \frac{-d\cdot{\rm{ln}}\left(P^{({\rm{T}})}_{{f}}\right)}{{\rm{ln}}\left(1+ez\right)\left[W_{-1}\left(z\right)+1\right]}. 
    \label{eq:thirdproof}
\end{equation}

\noindent Since $P^{({\rm{T}})}_{{f}}\ll 1$, the numerator of (\ref{eq:thirdproof}) is a positive number, and (\ref{eq:thirdproof}) is minimized when the denominator $g(z)={\rm{ln}}\left(1+ez\right)\left[W_{-1}\left(z\right)+1\right]$ is maximized. $g(z)$ is a convex function, with a maximum at $z\approx-0.2798$. We therefore obtain that $N_{RU}$ is minimized when 

\begin{equation}
n\approx-0.6995\cdot {\rm{ln}}\left(P^{({\rm{T}})}_{{f}}\right).
\label{eq:thirdproof2}
\end{equation}

\noindent By applying (\ref{eq:thirdproof2}) in (\ref{eq:secondproof0}), we obtain (\ref{eq:thirdproof0}).
\end{proof}

$N^{min}_{RU}$ has also a linear dependency with the number of devices, and a logarithmic dependency with the target failure probability. By comparing (\ref{eq:thirdproof0}) with (\ref{eq:nrumin2}), we can observe that the capability of the receiver to resolve one collision leads to a reduction of the minimum amount of resource units of a factor of $\sim\times 2.77$ for achieving the same target failure probability.

\section{Validation of analytical results}
We validate our analytical results with simulation studies, assuming a failure probability target $P^{({\rm{T}})}_{{f}}=10^{-6}$. Analytical results are obtained with (\ref{eq:nru}) and (\ref{eq:secondproof0}), for $N^{\rm{max}}_{\text{c}}=0$ and $N^{\rm{max}}_{\text{c}}=1$, respectively, and the analytical minima with (\ref{eq:nrumin}) and (\ref{eq:thirdproof0}). For the cases where analytical results are not available, i.e., $N^{\rm{max}}_{\text{c}}=2,3$, we resort to numerical results, obtained by numerical inversion of (\ref{eq:failureNc}) with respect to $N_{RU}$. Simulation results are obtained as follows: we first set a given number of packet repetitions $n$, with $n=2, \cdots 26$; we then set a value for $N_{RU}$, and generate up to 20 million samples of transmission patterns for the target and interfering devices; we calculate the empirical empirical failure probability when up to $N^{\rm{max}}_{\text{c}}=0,\cdots,3$ collisions are resolvable; we repeat the pattern generation for a large set of $N_{RU}$ values, and save, for the given $n$, the minimum $N_{RU}$ value that allows to achieve the failure probability target. 
Note that, for each $N_{RU}$, the resource grid is generated by selecting arbitrarily $p$ and $q$ such that $p,q \geq n$ and $pq=N_{RU}$; patterns of both target and interfering devices are then generated by randomly selecting $n$ resource units in the grid, with uniform probability.  

Fig. \ref{figres1} shows the amount of radio resource units needed to support the failure probability target with $d=100$ interfering devices. Simulation results show a perfect match with the analytical and numerical ones. When the receiver has no collision solving capabilities, a large number of packet repetitions is needed in order to have the highest chance of a collision-free transmission. The capability of resolving collisions leads to a lower amount of needed resource units, but also to a lower number of packet repetitions for approaching the minimum, reasonably translating to major transmit power reduction. Observe that, while for $N^{max}_{\text{c}}=0$ the number of required resource units becomes soon fairly insensitive to the number of packet repetitions, the $N^{max}_{\text{c}}\geq 1$ cases exhibit the intuitive trade-off between resource usage and generated interference, i.e., one needs a certain number of repetitions to improve resilience to collisions, but when such number increases, performance is affected by the higher likelihood of collisions themselves. This stresses the need of an accurate setting of the $n$ parameter in a practical system.

Fig. \ref{figres2} shows the accuracy of the analytical results in (\ref{eq:nrumin2}) and (\ref{eq:thirdproof0}) as a function of number of interfering devices $d$. Since the values of $n$ that lead to minimum resource usage in Fig. \ref{figres1} (orange circles) are independent from the number of interfering devices, they can be used for the analysis in Fig. \ref{figres2}. For $N^{\rm{max}}_{\text{c}}=0$, the analytical results are compared to the values obtained with (\ref{eq:nru}), using the value $n$ taken from the result in Fig. \ref{figres1}. With $N^{\rm{max}}_{\text{c}}=1$, the analytical results are compared to the numerical results obtained by numerical searches for solutions of (\ref{eq:secondproof}), with $n$ also taken from the result in Fig. \ref{figres1}. Though results in (\ref{eq:nrumin2}) and (\ref{eq:thirdproof0}) were obtained in the limit for large $d$, the match with the numerical results is nearly perfect also for a low number of interfering devices. 


\begin{figure}
    \centering
    \includegraphics[scale=0.46]{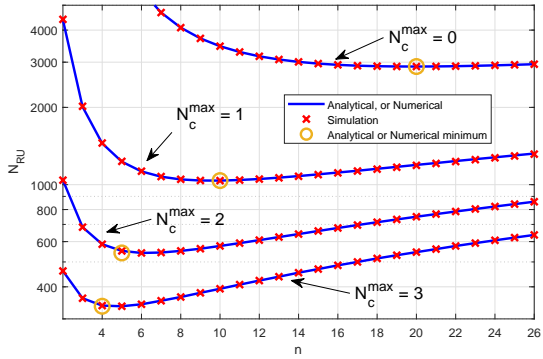}
    \caption{Number of resource units needed for supporting a target failure probability of $10^{-6}$, for 100 interfering devices.}
    \label{figres1}
\end{figure}


\begin{figure}
    \centering
    \includegraphics[scale=0.46]{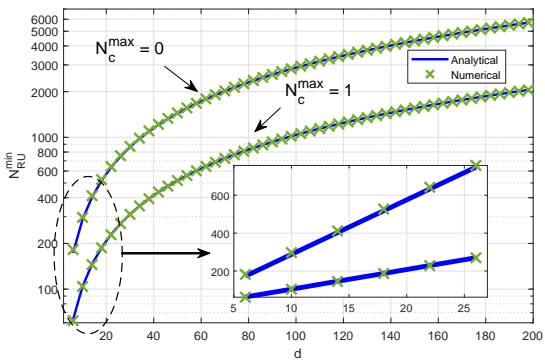}
    \caption{Accuracy of analytical results for minimum number of radio resource elements, assuming a target failure probability of $10^{-6}$.}
    \label{figres2}
\end{figure}


\section{Conclusions}
In this letter, we have studied the problem of dimensioning radio resources for a frequency-hopped radio system targeting ultra-reliable communication in the presence of a large number of persistent interferers. We have derived a closed form expression for the required number of repetitions and the minimum amount of resource units needed to achieve a failure probability target, considering the cases of collisions as disruptive events, as well as receiver tolerance to single collision events. The number of packet repetitions ensuring minimum radio resource usage results to be independent from the number of interfering devices; while the minimum number of resource elements grows linearly with the number of devices, and has a logarithmic dependency with the failure probability target. Numerical results have proved the accuracy of the derived analytical results. 

We believe that our results can provide concrete insights to wireless engineers for dimensioning frequency-hopped wireless protocols for high reliability, and can be used as benchmark for system simulation studies based on realistic models of radio propagation and receiver capabilities. 

\section*{Acknowledgements}
This research is partially supported by the HORIZON-JU-SNS-2022-STREAM-B-01-03 6G-SHINE project (Grant Agreement No. 101095738). 

\bibliographystyle{IEEEtran}
\bibliography{IEEEabrv,bibliography.bib}

\end{document}